\documentclass[11pt,a4paper]{article}

\usepackage{amsfonts,amsmath,amssymb,amsthm,bm}
\usepackage{latexsym,amscd}
\usepackage{amsbsy}
\usepackage{CJK}
\usepackage{indentfirst,mathrsfs}
\usepackage{latexsym,amscd}
\usepackage{amsbsy}
\usepackage[noadjust]{cite}
\usepackage{color}
\usepackage{multirow}
\usepackage{makecell}
\usepackage{float}
\usepackage{booktabs}
\newtheorem{theorem}{Theorem}
\usepackage{blkarray}
\usepackage{mathdots}

\newtheorem{lem}{Lemma}
\newtheorem{cor}{Corollary}

\newtheorem{con}{Construction}
\newtheorem{definition}{Definition}

\newtheorem{remark}{Remark}

\newcommand{\C}{{\mathcal{C}}}

\begin{document}

\title{\bf Strong Singleton-Like Bounds, Quasi-Perfect Codes and Distance-Optimal Codes in the Sum-Rank Metric}
\author{Chao Liu, Hao Chen, Qinqin Ji, Ziyan Xie, Dabin Zheng and Yongbo Xia\thanks{Chao Liu and Dabin Zheng are with the Hubei Laboratory of Applied Mathematics, Faculty of Mathematics, Hubei University, Wuhan 430062, Hubei, China (e-mail: chliuu@163.com, dzheng@hubu.edu.cn). Hao Chen is with the College of Information Science and Technology, Jinan University, Guangzhou, Guangdong Province, 510632, China (e-mail: haochen@jnu.edu.cn). Qinqin Ji is with the School of Mathematics and Statistics, Hubei University of Education, Wuhan 430205, Hubei, China (e-mail: qqinji@163.com). Ziyan Xie is with the School of Mathematics, Nanjing University of Aeronautics and Astronautics, Nanjiang 211106, China (e-mail: xieziyan@nuaa.edu.cn). Yongbo Xia is with School of Mathematics and Statistics, South-Central Minzu University, Wuhan, 430074, Hubei, China (e-mail:xia@mail.scuec.edu.cn). This research was supported by NSFC Grant 62032009, NSFC Grant 62272148, NSFC Grant 62171479, and Natural Science Foundation of Hubei Province of China Grant 2023AFB847.}}

\date{}
\maketitle

\begin{abstract}
Codes in the sum-rank metric have received many attentions in recent years, since they have wide applications in the multishot network coding, the space-time coding and the distributed storage. In this paper, by constructing covering codes in the sum-rank metric from covering codes in the Hamming metric, we derive new upper bounds on sizes, the covering radii and the block length functions of codes in the sum-rank metric. As applications, we present  several strong Singleton-like bounds that are tighter than the classical Singleton-like bound when block lengths are large. In addition, we give the explicit constructions of the distance-optimal sum-rank  codes of matrix sizes $s\times s$ and $2\times 2$ with minimum sum-rank distance four respectively by using cyclic codes in the Hamming metric.  More importantly, we  present an infinite families of quasi-perfect $q$-ary sum-rank codes with matrix sizes $2\times m$.  Furthermore, we construct almost MSRD codes with larger block lengths and demonstrate how the Plotkin sum can be used to give more distance-optimal sum-rank codes.
\end{abstract}

\vskip 6pt
\noindent {\bf Intex terms:} Covering code in the sum-rank metric, Strong Singleton-like bound, Distance-optimal code, Quasi-perfect code in the sum-rank metric.
\vskip6pt

\vskip 30pt

\section{Introduction}

\subsection{Codes in the Hamming metric and the sum-rank metric}

In this subsection, we recall some basic concepts on error-correcting codes in the Hamming metric and the sum-rank metric. Let ${\mathbb F_q}$ be a finite field with $q$ elements, where $q$ is a prime power.  For a vector ${\bf a}=(a_1,a_2,\dots, a_n)\in {\mathbb F}_q^n$, its Hamming weight $wt_H({\bf a})$ is the cardinality of its support:
\[
supp({\bf a})=\{i\,:\, a_i\neq 0\}.
\]
The Hamming distance $d_H({\bf a,b})$ between $\bf a$ and $\bf b$ is defined as $wt_H({\bf a}-{\bf b})$. An $[n,k,d_H]_q$ linear code $\C$ over ${\mathbb F_q}$  is a $k$-dimensional subspace of ${\mathbb F_q^n}$ with minimum Hamming distance $d_H$, where
$$d_H={\rm min}\left\{d_H({\bf u},{\bf v})\, :\, {\rm for \ all}\ {\bf u}\neq {\bf v}\in \C\right\}.$$
The minimum distance is bounded by the  Singleton Bound, i.e., $d_H\leq n-k+1$, see \cite{Singleton1964}. A linear code $\C$ is called a maximum distance separable (MDS) code if and only if $d_H(\C)=n-k+1$. MDS codes play an important role in cryptography and coding theory and have attracted lots of attention, involving the construction and non-equivalence.  Reed-Solomon (RS) codes are a class of special MDS codes have been extensively studied, see \cite{Huffman2003,Reed1960}.

For a code $\C$ in the Hamming metric space ${\mathbb F}_q^n$, we define its covering radius by
\[
R_H(\C)=\underset{{\bf x}\in {\mathbb{F}_q^n}}{\rm max}\,  \underset{\bf c\in \C}{\rm min}\{d_H({\bf x},{\bf c})\}.
\]
Then the Hamming balls
\[
B({\bf c},R_H(\C))=\{{\bf x}\in {\mathbb{F}_q^n}\,:\, d_H({\bf x},{\bf c})\leq R_H(\C)\}
\]
centered at all codewords ${\bf c}\in \C$ with radius $R_H(\C)$ cover the whole space ${\mathbb{F}_q^n}$, moreover, this radius is the smallest possible radius. Then $\C$ is called a covering code with radius $R_H(\C)$. We refer to the book \cite{Cohen1997} on this classical topic in coding theory. If a code $\C$ satisfies $R_H(\C)=\left\lfloor \frac{d_H(\C)-1}{2}\right\rfloor$, then $\C$ is called the perfect code \cite{MacWilliams1977,Vanlint1999}. If $R_H(\C)=\left\lfloor \frac{d_H(\C)-1}{2}\right\rfloor+1$, then $\C$ is called the quasi-perfect code \cite{MacWilliams1977,Etzion2005}. It was proved that perfect codes have the same parameters as Hamming codes and Golay codes, see \cite{Cohen1997,Vanlint1999}. Quasi-perfect codes are  ideal candidates  with which there is no perfect code, see \cite{Etzion2005}.

For a $(n,M,d)_q$ codes, the Sphere packing bound \cite{MacWilliams1977} asserts that
\[
MV_H(q,\left\lfloor \frac{d-1}{2} \right\rfloor )\leq q^n,
\]
where $V_H(q,r)=\sum\limits_{i=0}^r\left(\begin{array}{c}
    n \\
    i
\end{array}\right)(q-1)^i$ denotes the volume of the ball with the radius $r$ in the Hamming metric space ${\mathbb{F}_q^n}$. If there is a $(n,M,d)_q$ code, and there exists no code $(n,M,d+1)_q$, then code $(n,M,d)_q$ is called distance-optimal.

Codes in the sum-rank metric are the generalizations  of codes in the Hamming metric and codes in the rank metric. They are widely applied in network coding, see \cite{Napp2018}, space-time coding, see \cite{Lu2005,Shehadeh2022}, and coding for distributed storage, see \cite{Cai2022}. For fundamental properties and constructions of sum-rank codes, we refer to \cite{Byrne2021,Byrne2022,Camp2022,MP24,Ott2021}. Now, we recall some basic concepts and results on sum-rank codes. Let ${\mathbb F}_q^{(n,m)}$ be the set of all $n \times m$ matrices over ${\mathbb F}_q$, this is a linear space over ${\mathbb{F}_q}$ of the dimension $nm$. Let $n_i\leq m_i$, $i=1,2,\dots,t$ be $2t$ positive integers with $m_1\geq m_2\geq \cdots \geq m_t$,
\[
{\mathbb{F}_q^{(n_1,m_1),(n_2,m_2),\dots,(n_t,m_t)}}={\mathbb{F}_q^{(n_1,m_1)}}\oplus {\mathbb{F}_q^{(n_2,m_2)}}\oplus \cdots \oplus {\mathbb{F}_q^{(n_t,m_t)}}
\]
be the set of all ${\bf x}=({\bf x}_1,{\bf x}_2,\dots,{\bf x}_t)$, where ${\bf x}_i \in {\mathbb{F}_q^{(n_i,m_i)}}$, $i=1,2,\dots,t$, this is a linear space over ${\mathbb{F}_q}$ of dimension $\sum\limits_{i=1}^t n_im_i$. Its sum-rank weight is defined by
$$wt_{sr}({\bf x})={\rm rank}({\bf x}_1)+{\rm rank}({\bf x}_2)+\cdots+{\rm rank}({\bf x}_t),$$ and the sum-rank distance $d_{sr}({\bf x},{\bf y})$ between ${\bf x},\ {\bf y}$ is defined as $wt_{sr}({\bf x}-{\bf y})$ for ${\bf x},\ {\bf y}\in {\mathbb{F}_q^{(n_1,m_1),\dots,(n_t,m_t)}}$.

\begin{definition}
    A $q$-ary sum-rank code $\C \subset {\mathbb{F}_q^{(n_1,m_1),\dots,(n_t,m_t)}} $ with block length $t$ and matrix sizes $n_i\times m_i,\ i=1,\dots, t$ is a subset of the space ${\mathbb{F}_q^{(n_1,m_1),\dots,(n_t,m_t)}}$, its minimum sum-rank distance is defined by
    \[
    d_{sr}(\C)={\rm min}\left\{d_{sr}({\bf x},{\bf y})\,:\, {\rm for\  all}\ {\bf x}\neq {\bf y}\in \C \right\}.
    \]
\end{definition}
It is clear that the sum-rank metric code $\C$ is a Hamming metric error-correcting code if $n_i=m_i=1,\ i=1,2,\dots,t$. The basic goal of coding in the finite metric space ${\mathbb{F}_q^{(n_1,m_1),\dots,(n_t,m_t)}}$ endowed with sum-rank metric is to construct sum-rank codes with large sizes and large sum-rank distances. For some basic upper bounds, we refer to \cite{Byrne2021}.

The Singleton-like bound for sum-rank codes was proved in \cite{Byrne2021,Martinez2018}. If minimum sum-rank distance $d_{sr}$ can be  written uniquely as the form $d_{sr}=\sum\limits_{i=1}^{j-1}n_i+\delta+1$, where $0\leq \delta\leq n_j-1$ for $1\leq j\leq t$. Then
\[
|\C|\leq q^{\sum\limits_{i=j}^tn_im_i-m_j\delta}.
\]
The sum-rank code attaining this bound is called the maximum sum-rank distance (MSRD) code, when $m_1=m_2=\cdots=m_t=m$, this bound can be written by
\[
|\C|\leq q^{m(N-d_{sr}+1)},
\]
where $N=\sum\limits_{i=1}^tn_i$. When $n_i=m_i=1, \ i=1,\dots,t$, it degenerates to the original Singleton bound for the Hamming metric
codes. Similar to the case in the Hamming metric, the difference $m(N-d_{sr}+1)-\log_q{|\C|}$ is called the Singleton defect of sum-rank code $\C$. A code with the Singleton defect at most two is called almost MSRD code.

We have the following definition of covering radius in the sum-rank metric.
\begin{definition}
    Let $\C\subset {\mathbb{F}_q^{(n_1,m_1),\dots,(n_t,m_t)}}$ be a code in the sum-rank metric, its covering radius $R_{sr}(\C)$ is the minimum radius such that the balls
    \[
    B({\bf c},R_{sr}(\C))=\{{\bf x}\in {\mathbb{F}_q^{(n_1,m_1),\dots,(n_t,m_t)}}\,:\, d_{sr}({\bf x},{\bf c})\leq R_{sr}(\C)\}
    \]
    centered at all codewords ${\bf c}\in \C$ with radius $R_{sr}(\C)$ cover the whole space ${\mathbb{F}_q^{(n_1,m_1),\dots,(n_t,m_t)}}$.
\end{definition}

Let $V_{sr}(q,r)$ be the volume of the ball $\{{\bf x}\in {\mathbb{F}_q^{(n_1,m_1),\dots,(n_t,m_t)}}\,:\, d_{sr}({\bf x},{\bf 0})\leq r\}$ in the sum-rank metric space ${\mathbb{F}_q^{(n_1,m_1),\dots,(n_t,m_t)}}$. For a sum-rank code $\C$ with $M$ codewords and minimum sum-rank distance $d_{sr}$, the Sphere packing bound in the sum-rank metric \cite{Byrne2021} asserts  that
\[
MV_{sr}(q,\left\lfloor \frac{d_{sr}-1}{2} \right\rfloor) \leq q^{\sum\limits_{i=1}^tn_im_i}.
\]
A sum-rank code in ${\mathbb{F}_q^{(n_1,m_1),\dots,(n_t,m_t)}}$ with $M$ codewords and minimum sum-rank distance $d_{sr}$ is called the distance-optimal if there exists no sum-rank code with $M$ codewords and minimum sum-rank distance $d_{sr}+1$. If this code $\C$ satisfies
\[
MV_{sr}(q,\left\lfloor \frac{d_{sr}-1}{2} \right\rfloor) = q^{\sum\limits_{i=1}^tn_im_i},
\]
then this code is called perfect sum-rank code, and the whole space is the disjoint union of balls with the radius $\left\lfloor \frac{d_{sr}-1}{2} \right\rfloor$ centered at all codewords of a perfect code $\C$, i.e., $R_{sr}(\C)=\left\lfloor \frac{d_{sr}-1}{2} \right\rfloor$. If $R_{sr}(\C)=\left\lfloor \frac{d_{sr}-1}{2} \right\rfloor+1$, then code $\C$ is called quasi-perfect. Quasi-perfect sum-rank codes are naturally extensions of quasi-perfect codes of the matrix size $1\times 1$ in the Hamming metric. Then how to construct quasi-perfect sum-rank codes with the matrix size $n_i\times m_i, \ i=1,\dots,t$, where both $n_i$ and $m_i$ are larger than 1, is an interesting and challenging problem.

\subsection{Related works}
Sum-rank codes defined in the finite metric space $\mathbb{F}_{q}^{(n_1,m_1),\ldots,(n_t,m_t)}$ are the generalization of both codes in the Hamming-metric ($n_i = m_i = 1$) and codes in the rank-metric codes ($t = 1$). Fundamental properties and bounds, including the Singleton-like bound for sum-rank codes, were established in \cite{Byrne2021,Martinez2018}. For parameters satisfying $ t \leq q-1 $ and $ N \leq (q-1)m $, explicit MSRD constructions known as linearized Reed-Solomon codes were given in \cite{Martinez2018,Neri2022}. Conversely, it was shown in \cite[Example VI.9]{Byrne2021} that MSRD codes may not exist for some minimum distances when $t = q$. Further constructions include  linear MSRD codes over smaller fields \cite{MP24} and  linear MSRD codes with arbitrary block lengths and various square matrix sizes over a fixed field in \cite{Chen2023}. One-weight codes in the sum-rank metric were constructed in \cite{Neri2023}. Anticode bounds for sum-rank codes, optimal codes and the generalized sum-rank weights were given in \cite{Byrne2022,Camp2022}. The Hartmann-Tzeng for cyclic-skew-cyclic sum-rank codes was given in \cite{Alfarano2022}. Eigenvalue bounds for sum-rank codes were developed in \cite{Abiad2023} and these bounds were used to prove the nonexistence of MSRD codes for certain parameters. Covering codes in the rank-metric were studied in \cite{Byrne2017}. For covering codes in the sum-rank metric , some upper and lower bounds on covering radius have been preliminarily studied in \cite{Ott2022}. However these bounds were often related to ball volumes and not given explicitly. The relation between covering codes in the sum-rank metric and sum-rank-$\rho$-saturating systems was studied in \cite{BBB}.\\

Perfect sum-rank codes with matrix size $1 \times n$ and minimum distance three were constructed in \cite{UMart2019}. Several infinite families of distance-optimal sum-rank codes with minimum distance four were given in \cite{yqchen}. The construction of quasi-perfect codes in the Hamming metric with various parameters has received many attentions. Many binary, ternary linear or nonlinear quasi-perfect codes and quasi-perfect linear codes over large fields were constructed and classified. In particular, short quasi-perfect codes over ${\mathbb{F}_q}$ with the minimum distance four were constructed in \cite{Giulietti2007}. It is direct to verify that quasi-perfect codes with even minimum distances are distance-optimal.

\subsection{Our contributions}

In this paper, we give a construction of covering codes in the sum-rank metric space ${\mathbb{F}_q^{(m,m),\dots,(m,m)}}$. Block length functions of covering codes in the sum-rank metric are introduced, as the generalization of the length functions introduced in \cite{Brualdi1989} for covering codes in the Hamming metric. Then we give a general upper bound on sizes of covering sum-rank codes of the matrix size $m\times m$ and an upper bound on block length functions.   As applications of these bounds, we construct some covering codes in the sum-rank metric from covering codes in the Hamming metric.  Then we give several strong Singleton-like bounds for the sum-rank codes, and our bound is much stronger than the Singleton-like bound under certain conditions.

We also construct infinitely many new families of distance-optimal $q$-ary cyclic codes with the minimum distance four. These families of distance-optimal  codes have new parameters, compared with previous distance-optimal codes constructions. Infinitely many families of distance-optimal $q$-ary cyclic sum-rank codes with minimum distance four are presented.  MSRD codes are sum-rank codes with the zero Singleton defect, these MSRD codes constructed in \cite{Neri2023,Martinez2018} have the block lengths at most $q-1$. Singleton defect two almost MSRD $q$-ary codes with the block length up to $q^2$ and the minimum sum-rank distance four are given. An infinite family of distance-optimal $q$-ary sum-rank codes with the block length $q^4-1$, the matrix size $2\times 2$ and the minimum sum-rank distance four is constructed. The Singleton defect of these distance-optimal sum-rank
codes is four. It is showed that Plotkin sums of sum-rank codes lead to more distance-optimal sum-rank codes. Quasi-perfect binary sum-rank codes of the matrix size $2\times 2$ and the minimum sum-rank distance four are constructed. Notice that these quasi-perfect sum-rank
codes are distance-optimal automatically.

The paper is organized as follows. In Section 2, a construction of covering codes in the sum-rank metric from several covering codes in the Hamming metric is given. In Section 3, strong Singleton-like bounds on sum-rank codes are proved. In Section 4, several infinite families of quasi-perfect codes in the sum-rank metric are constructed. In Section 5, distance-optimal and almost MSRD codes are presented. In Section 6, the Plotkin sum of sum-rank codes is introduced. Section 7 concludes the paper.

\section{A construction of covering codes in the sum-rank metric}

In this section, we give a constructions of covering codes in the sum-rank metric from codes in the Hamming metric. The matrix sizes are assumed to satisfy $n_i=m_i=m,\ i=1,\dots,t$. Take a fixed basis $\Omega$ of ${\mathbb{F}_{q^m}}$ over ${\mathbb{F}_q}$, for each $a \in {\mathbb{F}_{q^m}}$, the coordinate representation of $a$ with respect to this basis is ${\bf a}=(a_1,a_2,\dots,a_m),\ a_i\in {\mathbb{F}_q}$. We denote $M_j(a)$  the $m\times m$ matrix with the only nonzero $j$-th row ${\bf a}$. The key point of our construction of codes in the sum-rank metric is the identification
\[
{\mathbb{F}_q^{(m,m)}}=\{M_1(\alpha_1)+M_2(\alpha_2)+\cdots+M_m(\alpha_m)\,:\,\alpha_j\in {\mathbb{F}_{q^m}}\ {\rm for}\ 1\leq j\leq m \}.
\]
This is the isomorphism of the linear space over ${\mathbb{F}_q}$.
\begin{con}
Let $\C_1,\C_2,\dots,\C_m$ be $m$ general  Hamming metric codes in ${\mathbb{F}_{q^m}^t}$, for any codeword ${\bf c}_i\in \C_i,\ i=1,\dots,m$, we denote by ${\bf c}_{ij}$ the $j$-th position of codeword ${\bf c}_i$ for $1\leq j\leq t$. Then a sum-rank code can be given as follows:
\[
SR_{covering}(\C_1,\dots,\C_m)=\left\{{\bf c}=\left(\sum_{i=1}^mM_i({\bf c}_{i1}),\sum_{i=1}^mM_i({\bf c}_{i2}),\dots,\sum_{i=1}^mM_i({\bf c}_{it})\right)\,:\, {\bf c}_1\in \C_1,\dots,{\bf c}_m\in \C_m \right\}.
\]
It is easy to see that the size of this sum-rank code is $\prod\limits_{i=1}^m|\C_i|$.
\end{con}

In this section, we present some results concerning the upper bounds on the block length functions. We first give some fundamental results on the covering radius of codes.
\begin{theorem}
    Let $\C_1,\C_2,\dots,\C_m$ be $m$ general  Hamming metric codes in ${\mathbb{F}_{q^m}^t}$ with covering radius $R_1,R_2,\dots,R_m$ respectively. Then for any ${\bf x}\in {\mathbb{F}_q^{(m,m),\dots,(m,m)}}$, there is a codeword ${\bf c}\in SR_{covering}(\C_1,\\\C_2,\dots,\C_m)$ such that
    \[d_{sr}({\bf x},{\bf c})\leq R_1+R_2+\cdots+R_m. \]
\end{theorem}

\begin{proof}
    Let ${\bf x}_1,{\bf x}_2,\dots,{\bf x}_m$ be any $m$ vectors in ${\mathbb{F}_{q^m}^t}$, then there exist ${\bf c}_1\in \C_1, {\bf c}_2\in \C_2,\dots,{\bf c}_m\in \C_m$ such that $wt_H({\bf x}_i- {\bf c}_i)\leq R_i,\ i=1,\dots,m$. Let
    \[S_i=supp({\bf x}_i-{\bf c}_i)=\{j\,:\, ({\bf x}_i-{\bf c}_i)_j\neq0\}.\]
    It is known that any ${\bf x}\in {\mathbb{F}_q^{(m,m),\dots,(m,m)}}$ can be expressed as
    \[
    {\bf x}=\left(\sum_{i=1}^mM_i({\bf x}_{i1}),\sum_{i=1}^mM_i({\bf x}_{i2}),\dots,\sum_{i=1}^mM_i({\bf x}_{it})\right),
    \]
    and set
    \[
    {\bf c}=\left(\sum_{i=1}^mM_i({\bf c}_{i1}),\sum_{i=1}^mM_i({\bf c}_{i2}),\dots,\sum_{i=1}^mM_i({\bf c}_{it})\right).
    \]
    For $1\leq r\leq t$, set $M_r=\sum\limits_{i=1}^m\left(M_i({\bf x}_{ir})-M_i({\bf c}_{ir})\right)$ and denote  the number of nonzero rows in matrix $M_r$  by $m_r$. Then
    \[wt_{sr}({\bf x}-{\bf c})=\sum_{r=1}^t{\rm rank}\left(\sum_{i=1}^m\left(M_i({\bf x}_{ir})-M_i({\bf c}_{ir})\right)\right)\leq \sum_{r=1}^tm_r=\sum_{r=1}^t\sum_{i=1}^m1_{r\in S_i}=\sum_{i=1}^m|S_i|\leq \sum_{i=1}^mR_i,
    \]
    where $1_{r\in S_i}=1$ if $r\in S_i$, otherwise, $1_{r\in S_i}=0$.
\end{proof}

In particular, if all codes in the Hamming metric are the same, we obtain the following corollary.
\begin{cor}\label{upper bound of radius of SR(C)}
     Let $\C$ be a code in ${\mathbb{F}_{q^m}^t}$ in the Hamming metric with covering radius $R$. Then the code $SR_{covering}(\C,\dots,\C)$ is a covering code in the sum-rank metric with covering radius at most $mR$.
\end{cor}

We define the block length functions $\ell_{q,m}(r,R)$ for linear covering codes in the sum-rank metric space ${\mathbb{F}_q^{(m,m),\dots,(m,m)}}$ with $t$ blocks. It is the smallest block length $t$ such that  there exists  sum-rank covering code in ${\mathbb{F}_q^{(m,m),\dots,(m,m)}}$   with block length $\ell_{q,m}(r,R)$, the codimension $r$ and the radius at least $R$. Block length functions $\ell_{q,m}(r,R)$ are the generalization of length functions $\ell_{q}(r,R)$ in the Hamming metric, which denote the smallest length  such that there is  a linear code with length $\ell_q(r,R)$, codimension $r$ and covering radius at least $R$. Then we have the following results on the block length functions.

\begin{lem}\label{relation of length function in Hamming metric and sum-rank metric}
    Let $q$ be a prime power, $m$, $r$, $R$ be fixed positive integers satisfying $m^2|r$ and $m|R$. Then
    \[\ell_{q,m}(r,R)\leq \ell_{q^m}(\frac{r}{m^2},\frac{R}{m}).\]
\end{lem}

\begin{proof}
    Let $\C$ be a covering code over ${\mathbb{F}_{q^m}}$ with length $\ell_{q^m}(\frac{r}{m^2},\frac{R}{m})$. Then $\C'=SR_{covering}(\C,\dots,\C)\subset {\mathbb{F}_q^{(m,m),\dots,(m,m)}}$  is a sum-rank code with block length $\ell_{q^m}(\frac{r}{m^2},\frac{R}{m})$, and we have
    \[{\rm dim}_{\mathbb{F}_{q^m}}(\C)=\ell_{q^m}(\frac{r}{m^2},\frac{R}{m})-\frac{r}{m^2},\ {\rm and } \ R_{sr}(SR_{covering}(\C,\dots,\C))\leq m\cdot \frac{R}{m}=R.\]
    Then the codimension of $\C'$ is
    \[
    {\rm Codim}_{\mathbb{F}_q}(SR_{covering}(\C,\dots,\C))=m^2\ell_{q^m}(\frac{r}{m^2},\frac{R}{m})-m{\rm dim }_{\mathbb{F}_q}(\C)=r.
    \]
    Since $\ell_{q,m}(r,R)$ is the smallest block length of sum-rank code with   codimension $r$ and radius  $R$, we have $\ell_{q,m}(r,R)\leq \ell_{q^m}(\frac{r}{m^2},\frac{R}{m})$.

\end{proof}

\begin{lem}\cite[Theorem 3.3]{Davydov2022}\label{upper bound of length function in the Hamming metric}
    Let $q$ be a prime power, $r$, $R$ be a positive integer. Then
    \[ \ell_{q}(r,R)\leq cq^{\frac{r-R}{R}}\cdot \sqrt[R]{\ln q}, \ {\rm  for}\ R\geq 3,\ r=tR+1,\ t\geq 1,  \]
    where $c$ is a universal constants independent of $q$ and $m$.
\end{lem}

\begin{theorem}\label{upper bound of length function in the sum-rank metric}
    Let $q$ be a prime power, $m$, $u$, $R$ be  positive integers satisfying $m^2|u$ and $m|R$. Then
    \[
    \ell_{q,m}(uR+m^2,R)\leq cq^{\frac{(u-m)R+m^2}{R}}\cdot (m\ln q)^{\frac{m}{R}},
    \]
    where $c$ is a universal constants independent of $q$ and $m$.
\end{theorem}

\begin{proof}
    By Lemma \ref{relation of length function in Hamming metric and sum-rank metric} and Lemma \ref{upper bound of length function in the Hamming metric}, we have
    \[
    \ell_{q,m}(uR+m^2,R)\leq \ell_{q^m}(\frac{uR+m^2}{m^2},\frac{R}{m})\leq cq^{m \cdot \frac{\frac{uR+m^2}{m^2}-\frac{R}{m}}{\frac{R}{m}}}\cdot (m\ln q)^{\frac{m}{R}}=cq^{\frac{(u-m)R+m^2}{R}}\cdot (m\ln q)^{\frac{m}{R}}.
    \]
\end{proof}

Then we give the bounds on the sizes of sum-rank codes. Let $q$ be a fixed prime power, $m$, $t$ be  given positive integers, for a given positive integer $R<mt$, we denote by $K_{q,m}(t, R)$ and $K_{q^m}(t,R)$ the minimum size of sum-rank codes $\C\subset {\mathbb{F}_q^{(m,m),\dots,(m,m)}}$ with block length $t$, covering radius smaller than or equal to $R$ and the minimum size of   codes $\C\subset{\mathbb{F}_{q^m}^t}$ with covering radius smaller than or equal to $R$ respectively. Set
\[k_t(q,\rho,m)=\frac{\log_qK_{q,m}(t,\rho mt)}{m^2t}, \ k_t(q^m,\rho)=\frac{\log_{q^m}K_{q^m}(t,\rho t)}{t}.\]
The following bound is well known,
\[
1-H_{q^m}(\rho)\leq k_t(q^m, \rho)\leq 1-H_{q^m}(\rho)+O(\frac{\log t}{t}),
\]
where
\[H_{q^m}(\rho)=\rho \log_{q^m}(q^m-1)-\rho \log_{q^m}\rho-(1-\rho)\log_{q^m}(1-\rho)\]
is the $q$-ary entropy function, see \cite{Cohen1997}. Then we have following upper bound on the sizes of sum-rank codes.
\begin{cor}\label{size of sum-rank codes}
    With notations defined as above and $m|R$, we have
    \[K_{q,m}(t,R)\leq \left(K_{q^m}(t,\frac{R}{m})\right)^m,\ {\rm and }\ k_t(q,\rho, m)\leq 1-H_{q^m}(\rho)+O(\frac{\log t}{t}).\]
\end{cor}

\begin{proof}
    Let $\C'\subset {\mathbb{F}_{q^m}^t}$ be a code with size $K_{q^m}(t,\frac{R}{m})$. Then $SR_{covering}(\C',\dots,\C')\subset {\mathbb{F}_q^{(m,m),\dots,(m,m)}}$ is a sum-rank code with block length $t$, and we have
    \[
    R_{sr}(SR_{covering}(\C',\dots,\C'))\leq m \cdot \frac{R}{m}=R, \ K_{q,m}(t,R)\leq |SR_{covering}(\C',\dots,\C')|=\left(K_{q^m}(t,\frac{R}{m})\right)^m,
    \]
    and
    \[
    k_t(q,\rho,m)=\frac{\log_qK_{q,m}(t,\rho mt)}{m^2t}\leq \frac{m^2\log_{q^m}K_{q^m}(t,\rho t)}{m^2t}=k_t(q^m,\rho)\leq 1-H_{q^m}(\rho)+O(\frac{\log t}{t}).
    \]
\end{proof}

\section{Strong Singleton-like bounds on general sum-rank codes}

In this section, we present some strong Singleton-like bounds on sum-rank codes, these bounds are stronger than the Singleton-like bound, when block lengths are large. We first have the following simple Lemma.

\begin{lem}\label{covering radius of extended sum-rank code}
    Let $\C'\subset {\mathbb{F}}_q^{(m,m),\dots,(m,m)}$ be a sum-rank code with block length $t$ and  covering radius $R$. Then an extended sum-rank code $\C=\C'\bigoplus \left({\mathbb{F}_q^{(m,m)}}\right)^{n-t}$ for $n>t$  is a sum-rank code with block length $n$ and covering radius $R$, where $\left({\mathbb{F}_q^{(m,m)}}\right)^{n-t}$ denotes $\bigoplus\limits_{i=1}^{n-t}{\mathbb{F}_q^{(m,m)}}$.
\end{lem}

\begin{proof}
    For any ${\bf x}=({\bf x}_1,{\bf x}_2,\dots,{\bf x}_n)\in \left({\mathbb{F}_q^{(m,m)}}\right)^n$, there exists a codeword ${\bf c}=({\bf c}_1,\dots,{\bf c}_t,{\bf x}_{t+1},\dots,{\bf x}_n)\\\in \C$ such that
    \[
    d_{sr}({\bf x},{\bf c})=\sum_{i=1}^t{\rm rank}({\bf x}_i-{\bf c}_i)+\sum_{j=t+1}^n{\rm rank}({\bf x}_j-{\bf x}_j)\leq R+0=R.
    \]
    Then we have $R_{sr}(\C)\leq R$.  Since the covering radius of $\C'$ is $R$,  we have
    \[
    R=\underset{{\bf y}\in \left({\mathbb{F}_q^{(m,m)}}\right)^{t}}{\rm max}\,  \underset{{\bf c}'\in \C'}{\rm min}\{d_{sr}({\bf y},{\bf c}')\}.
    \]
    Then there exists ${\bf y}=({\bf y}_1,{\bf y}_2,\dots,{\bf y}_t)\in \left({\mathbb{F}_q^{(m,m)}}\right)^t$ such that
    \[
    \underset{{\bf c}'\in \C'}{\rm min}\{d_{sr}({\bf y},{\bf c}')\}=\underset{{\bf c}'\in \C'}{\rm min}\left\{\sum_{i=1}^t{\rm rank}({\bf y}_i-{\bf c}'_i)\right\}=R.
    \]
    Therefore for ${\bf x}=({\bf y}_1,\dots,{\bf y}_t,{\bf x}_{t+1},\dots,{\bf x}_n)\in \left({\mathbb{F}_q^{(m,m)}}\right)^n$, we have
    \[
     R_{sr}(\C) \geq  \underset{{\bf c}\in \C}{\rm min}\{d_{sr}({\bf x},{\bf c})\}=\underset{{\bf c}\in \C}{\rm min}\left\{\sum_{i=1}^t{\rm rank}({\bf y}_i-{\bf c}_i')+\sum_{j=t+1}^{n}{\rm rank}({\bf x}_j-{\bf c}_j)\right\}\geq R.
    \]
    Hence, $R_{sr}(\C)=R$.
\end{proof}

Then we give the strong Singleton-like bound by constructing a sum-rank covering code from binary BCH codes. We denote the binary primitive  BCH codes of length $2^m-1$ and designed distance $2e+1$ by ${\rm BCH}(e,m)$, it is a cyclic code with parameters $[2^m-1,k\geq 2^m-me-1,d\geq 2e+1]$, see \cite{Cohen1997}. Moreover, we have the following result on the covering radius of this code.

\begin{lem}\cite[Theorem 10.3.1]{Cohen1997}\label{covering radius of BCH(e,n)}
If $q=2^m\geq (2e-1)^{4e+2}$, then
\[2e-1\leq R_H({\rm BCH}(e,m))\leq 2e.\]
\end{lem}

\begin{theorem}
    Let $\C\subset {\mathbb{F}_2^{(m,m),\dots,(m,m)}}$ be a binary sum-rank code with block length $t\geq 2^n-1$ and minimum sum-rank distance $d_{sr}$, e be a positive integer satisfying $2^n\geq (2e-1)^{4e+2}$. Then
    \begin{itemize}
        \item [\rm (1)] If $d_{sr}=4m^2e+i,\ i=1,2,\dots,4m^2-1$, then $|\C|\leq 2^{m^2(t-ne)}$.
        \item [\rm (2)] If $d_{sr}=4m^2e$, then $|\C|\leq 2^{m^2(t-ne+n)}$.
    \end{itemize}
\end{theorem}

\begin{proof}
    We consider a binary primitive BCH  code with parameters $[2^n-1,2^n-ne-1,2e+1]$ and covering radius $R$, denoted by ${\rm BCH}(e,n)$. Let
    \[{\rm BCH}(e,n)_{2^m}={\rm BCH}(e,n)\bigotimes {\mathbb{F}_{2^m}}\subset {\mathbb{F}_{2^m}^{2^n-1}},\]
    where $\bigotimes $ denotes tensor product, i.e., ${\rm BCH}(e,n)_{2^m}$ has parameters $[2^n-1,2^n-ne-1,2e+1]$ over ${\mathbb{F}_{2^m}}$ and  covering radius at most $mR$. By Lemma \ref{covering radius of BCH(e,n)}, $R_H({\rm BCH}(e,n)_{2^m})\leq 2me$, Then from Theorem 2.1, the sum-rank code $SR_{covering}({\rm BCH}(e,n)_{2^m},\dots,{\rm BCH}(e,n)_{2^m})\subset {\mathbb{F}_2^{(m,m),\dots,(m,m)}}$ with block length $2^n-1$ has cover radius at most $2m^2e$ and
    \[
        |SR_{covering}({\rm BCH}(e,n)_{2^m},\dots,{\rm BCH}(e,n)_{2^m})|=|{\rm BCH}(e,n)_{2^m}|^{m}=2^{m^2(2^n-ne-1)}.
    \]
    Let
    \[\C'=SR_{covering}({\rm BCH}(e,n)_{2^m},\dots,{\rm BCH}(e,n)_{2^m})\bigoplus \left({\mathbb{F}_q^{(m,m)}}\right)^{t-(2^n-1)},
    \]
    then by Lemma \ref{covering radius of extended sum-rank code}, we have \[R_{sr}(\C')\leq 2m^2e, \ {\rm and}\ |\C'|=2^{m^2(2^n-ne-1)}\cdot 2^{m^2(t-(2^n-1))}=2^{m^2(t-ne)}.
    \]
    If $d_{sr}=4m^2e+i\geq 4m^2e+1$, then we have each ball centered at a codeword of $\C'$ with radius $R_{sr}(\C')$ contains at most one codeword of $\C$. Moreover, all balls centered at all codewords of $\C'$ with radius $R_{sr}(\C')$ cover the whole space ${\mathbb{F}}_q^{(m,m),\dots,(m,m)}$ ($t$ blocks). Hence,
    \[|\C|\leq |\C'|=2^{m^2(t-ne)}.\]

    If $d_{sr}=4m^2e$. Similarly, we can construct  a sum-rank code $SR_{covering}({\rm BCH}(e-1,n)_{2^m},\dots,{\rm BCH}(e-1,n)_{2^m})\subset {\mathbb{F}_2^{(m,m),\dots,(m,m)}}$ with block length $2^n-1$ has cover radius at most $2m^2(e-1)$ and
    \[
        |SR_{covering}({\rm BCH}(e-1,n)_{2^m},\dots,{\rm BCH}(e-1,n)_{2^m})|=|{\rm BCH}(e-1,n)_{2^m}|^{m}=2^{m^2(2^n-n(e-1)-1)}.
    \]
    Let
    \[\C'=SR_{covering}({\rm BCH}(e-1,n)_{2^m},\dots,{\rm BCH}(e-1,n)_{2^m})\bigoplus \left({\mathbb{F}_q^{(m,m)}}\right)^{t-(2^n-1)},
    \]
    then  we have
    \[R_{sr}(\C')\leq 2m^2(e-1),\ {\rm and}\  |\C'|=2^{m^2(2^n-n(e-1)-1)}\cdot 2^{m^2(t-(2^n-1))}=2^{m^2(t-ne+n)}.
    \]
    If $d_{sr}=4m^2e> 2R_{sr}(\C')+1$, similarly, we have
    \[|\C|\leq |\C'|=2^{m^2(t-ne+n)}.\]
\end{proof}

\begin{remark}
    When block lengths are large, for example, $n\geq 4m^2+1$, the above strong Singleton-like bound is much tighter than Singleton-like bound for sum-rank codes. On the other hand, it is clear if smaller covering codes in the sum-rank metric can be constructed, the above bound can be improved.
\end{remark}

If  the minimum sum-rank distance $d_{sr}=2R+1$ is fixed, we have the following strong Singleton-like bound from the block length function.
\begin{theorem}
    Let $q$ be a prime power, $m$, $u$, $R$ be  positive integers satisfying $m^2|u$ and $m|R$, $\C\subset {\mathbb{F}_q^{(m,m),\dots,(m,m)}}$ be a sum-rank code with block length $t$ and minimum sum-rank distance $d_{sr}=2R+1$. If
    \[
    t\geq cq^{\frac{(u-m)R+m^2}{R}}\cdot (m\ln q)^{\frac{m}{R}},
    \]
    where  $c $ is a universal constants independent of $q$ and $m$, then $\C$ has at most $q^{(t-1)m^2-u\cdot \frac{d_{sr}-1}{2}}$ codewords.
\end{theorem}

\begin{proof}
    It is known that $cq^{\frac{(u-m)R+m^2}{R}}\cdot (m\ln q)^{\frac{m}{R}}\geq \ell_{q,m}(uR+m^2,R)$ by Theorem \ref{upper bound of length function in the sum-rank metric}. Let $\C_1\subset {\mathbb{F}_q^{(m,m),\dots,(m,m)}}$ be a sum-rank code with  block length $\ell_{q,m}(uR+m^2,R)$, and
    \[
    \C_2=\C_1\bigoplus \left({\mathbb{F}_q^{(m,m)}}\right)^{t-\ell_{q,m}(uR+m^2,R)}.
    \]
    Then we have
    ${\rm dim}_{\mathbb{F}_q}\C_2=m^2t-(uR+m^2)$, and $R_{sr}(\C_2)=R$. Since $d_{sr}=2R+1$, then we have each ball centered at a codeword of $\C_2$ with radius $R_{sr}(\C_2)=R$ contains at most one codeword of $\C$. Moreover, all balls centered at all codewords of $\C_2$ with radius $R$ cover the whole space ${\mathbb{F}}_q^{(m,m),\dots,(m,m)}$ ($t$ blocks). Hence,
    \[
    |\C|\leq |\C_2|=q^{m^2t-(uR+m^2)}=q^{(t-1)m^2-u\cdot \frac{d_{sr}-1}{2}}.
    \]
\end{proof}

\begin{remark}
    The above result asserts that if the block length is sufficiently larger than a function of $u$, and $u\geq 2m-\frac{m^2}{R}$, then we get a strong Singleton-like bound, which is tighter than the Singleton-like bound.
\end{remark}

\section{Constructions of quasi-perfect sum-rank codes}

In this section, we construct a class of quasi-perfect codes in the sum-rank metric with the matrix size $2\times m$. Before giving the explicit constructions of this family of quasi-perfect codes, we need the following result.

\begin{lem}\label{the number of rank 1}
    There are $\frac{(q^{s_1}-1)(q^{s_2}-1)}{q-1}$ matrices with rank one in $\mathbb{F}_q^{(s_1,s_2)}$.
\end{lem}

\begin{proof}
    Let $M_{s_1\times s_2}$ denote a matrix with rank one, then there exist column vectors ${\bf u}\in {\mathbb{F}_q^{s_1}}$ and ${\bf v}\in {\mathbb{F}_q^{s_2}}$ such that ${M_{s_1\times s_2}}={\bf u}{\bf v}^T$ with ${\bf u},{\bf v}\neq {\bf 0}$. Therefore, we have $q^{s_1}-1$ vectors ${\bf u}$ and  $q^{s_2}-1$ vectors ${\bf v}$. Since ${\bf u}{\bf v}^T=c{\bf u}(c^{-1}{\bf v})^T$ denotes the same matrix, where $c\in {\mathbb{F}_q^*}$. Hence, there are $\frac{(q^{s_1}-1)(q^{s_2}-1)}{q-1}$ matrices with rank one in $\mathbb{F}_q^{(s_1,s_2)}$.
\end{proof}

\begin{lem}\label{volume of ball with radius 2 in the sum-rank metric}
    The volume of the ball $V_{sr}(q,2)\subset {\mathbb{F}_q^{(s,s),\dots,(s,s)}}$ with $t$ blocks satisfies
    \[V_{sr}(q,2)\geq \frac{t(t-1)(q^s-1)^4}{2(q-1)^2}.\]
\end{lem}

\begin{proof}
    Let $V_{sr}(q,2)$ denote all vectors in ${\mathbb{F}_q^{(s,s),\dots,(s,s)}}$ with sum-rank weight less than or equal to 2. We consider any two positions of the vectors in ${\mathbb{F}_q^{(s,s),\dots,(s,s)}}$, where the matrices corresponding to these two positions have rank 1. Then by Lemma \ref{the number of rank 1}, we have
    \[
    V_{sr}(q,2)\geq \left(\begin{array}{c}
        t\\
        2
    \end{array}\right)\frac{(q^s-1)^2}{q-1}\frac{(q^s-1)^2}{q-1}=\frac{t(t-1)(q^s-1)^4}{2(q-1)^2}.
    \]
\end{proof}

We recall the construction of sum-rank codes in \cite{Chen2023}.

\begin{con}
    Let $\C_1,\C_2,\dots,\C_m$ be $m$ general  Hamming metric codes in ${\mathbb{F}_{q^m}^t}$, for any codeword ${\bf c}_i\in \C_i,\ i=1,\dots,m$, we denote ${\bf c}_{ij}$ the $j$-th position of codeword ${\bf c}_i$ for $1\leq j\leq t$. Then a sum-rank code can be given as follows:
\[
SR(\C_1,\dots,\C_m)=\left\{{\bf c}=\left(M_1,M_2,\dots,M_t\right) \right\},
\]
where $M_i$ denote the matrix corresponding to  $q$-polynomial $f_i(x)={\bf c}_{1i}x+{\bf c}_{2i}x^{q}+\cdots+{\bf c}_{mi}x^{q^{m-1}},\ i=1,\dots,t$. It is easy to see that the size of this sum-rank code is $\prod\limits_{i=1}^m|\C_i|$.
\end{con}

Let $\C_1,\C_2,\dots,\C_n\subset {\mathbb{F}_{q^m}^t}$ be $n$ linear codes. Then a sum-rank code over ${\mathbb{F}_q^{(n,m),\dots,(n,m)}}$ with block length $t$ is given as follows:
\[
SR(\C_1,\dots,\C_n)=\left\{{\bf c}=\left(M_1,M_2,\dots,M_t\right) \right\},
\]
where $M_i$ denote the matrix corresponding to  $q$-polynomial $f_i(x)={\bf c}_{1i}\phi(x)+{\bf c}_{2i}\phi(x^{q})+\cdots+{\bf c}_{ni}\phi(x^{q^{n-1}}),\ i=1,\dots,t$.

If we view a $q$-polynomial $f(x)=a_1x+a_2x^q+\cdots+a_mx^{q^{m-1}}$ as an ${\mathbb{F}_q}$-linear mapping, where $\ a_i\in {\mathbb{F}_{q^m}},\ i=1,2,\dots,m$, and fix a basis of ${\mathbb{F}_{q^m}}$ over ${\mathbb{F}_q}$, then this $q$-polynomial corresponds a $m\times m$ matrix over ${\mathbb{F}_q}$, and it is easy to verify that the rank of ${\mathbb{F}_q}$-linear mapping does not depend on the basis.

In the case of sum-rank codes over ${\mathbb{F}_{q}^{(n,m),\dots,(n,m)}}$, where $n<m$. We can identify the matrix space ${\mathbb{F}_q^{(n,m)}}$ as the space of all $q$-polynomials $a_1\phi(x)+a_2\phi(x^q)+\dots+a_{n}\phi(x^{q^{n-1}})$, where $x\in {\mathbb{F}_{q^n}}$, $a_i\in {\mathbb{F}}_{q^m}$ for $1\leq i\leq n $ and $\phi: {\mathbb{F}_{q^n}}\rightarrow {\mathbb{F}_{q^m}}$ is an ${\mathbb{F}}_q$-linear injective mapping.

\begin{lem}\cite[Theorem 2.1]{Chen2023}
    Let $\C_1,\C_2,\dots,\C_m\subset {\mathbb{F}_{q^m}^t}$ be $m$ linear codes with minimum Hamming distance $d_1,d_2,\dots,d_m$ respectively. Then sum-rank code $SR(\C_1,\C_2,\dots,\C_m)$ with block length $t$ has the minimum sum-rank distance
    \[d_{sr}(SR(\C_1,\dots,\C_m))\geq {\rm min} \{d_1,2d_2,\dots,md_m\}.
    \]
\end{lem}

\begin{theorem}
    Let $t=\frac{q^{mu}-1}{q^m-1}$, $u$ be a positive integer. Then an almost distance-optimal and quasi-perfect sum-rank code with block length $t$, matrix size $2\times m$ is constructed explicitly.
\end{theorem}

\begin{proof}
    Let $\C_1$ be a Hamming code with parameters $[t,t-u,3]_{q^m}$, $\C_2$ be a trivial cyclic code with parameters $[t,t-1,2]_{q^m}$. Then we have $SR(\C_1,\C_2)$ is a sum-rank code with block length $t$, matrix size $2\times m$, and
    \[
    d_{sr}(SR(\C_1,\C_2))\geq 3,\ {\rm dim}_{\mathbb{F}_q}(SR(\C_1,\C_2))=2mt-m(u+1).
    \]
    It is an almost distance-optimal since
    \[
    V_{sr}(q,2)\geq \frac{t(t-1)}{2}(q+1)^2(q^m-1)^2=\frac{(q^{mu}-1)(q^{mu}-q^m)}{2}(q+1)^2>q^{m(u+1)}.
    \]
    Then we prove that this sum-rank code is quasi-perfect. It is known that the Hamming code $\C_1$ is a perfect code in the Hamming metric, then we have $R_H(\C_1)=1$. Therefore, there exists a codeword ${\bf c}_1\in \C_1$, such that $wt_{H}({\bf v}_1-{\bf c}_1)=1$ for any vector ${\bf v}_1\in {\mathbb{F}_{q^m}^t}$. Without loss of generality, we assume that $j$-th position of vector ${\bf v}_1-{\bf c}_1$ is not zero, denoted by $({\bf v}_1-{\bf c}_1)_j\neq 0$. Moreover, $\C_2$ is a trivial code with radius one. Hence, for any vector ${\bf v}_2\in {\mathbb{F}_{q^m}^t}$ , there exists a codeword ${\bf c}_2\in \C_2$ such that $wt_H({\bf v}_2-{\bf c}_2)=1$, and the nonzero position can be at any position in $\{1,2,\dots,t\}$, we assume that $({\bf v}_2-{\bf c}_2)_j\neq 0$. Therefore, for any ${\bf v}=(M_1,M_2,\dots,M_t)\in {\mathbb{F}_q^{(2,m),\dots,(2,m)}}$, there exists a codeword ${\bf c}=(N_1,N_2,\dots,N_t)\in SR(\C_1,\C_2)$, where $M_i$ and $N_i$ denote the matrix corresponding to  $q$-polynomial
    \[f_i(x)={\bf v}_{1i}\phi(x)+{\bf v}_{2i}\phi(x^{q}),\ g_i(x)={\bf c}_{1i}\phi(x)+{\bf c}_{2i}\phi(x^{q}),  i=1,\dots,t\]
    respectively, such that ${\bf v}-{\bf c}$ has only one nonzero position, and we have \[wt_{sr}({\bf v}-{\bf c})\leq 2,\ R_{sr}(SR(\C_1,\C_2))\leq 2.\]
    Therefore, $SR(\C_1,\C_2)$ is a quasi-perfect sum-rank code.
\end{proof}

Before giving the quasi-perfect sum-rank codes with matrix size $2\times 2$,  we recall some results on quasi-perfect codes \cite{Giulietti2007} in the Hamming metric. Let
\[s_{m,q}=3(q^{\left\lfloor \frac{m-3}{2}\right\rfloor} +q^{\left\lfloor \frac{m-3}{2}\right\rfloor-1}+\cdots+q)+2,\ n=2q^{\frac{m-2}{2}}+s_{m,q},\]
$q$ be an even square, $m\geq 7$ be odd. An infinitely family of quasi-perfect codes with parameters $[n,n-m,4]_q$ were constructed, see \cite[Proposition 2.5]{Giulietti2007}. For other even prime power, infinitely family of quasi-perfect codes with parameters $[n,n-m,4]_q$ was constructed, see \cite[Proposition 4.1]{Giulietti2007}. It is obvious that quasi-perfect codes and sum-rank codes with even minimum distances and even minimum sum-rank distances are distance-optimal. Moreover, we have the following lemma for the binary sum-rank codes with matrix size $2\times 2$.

\begin{lem}\cite[Theorem 2.2]{Chen2025}\label{weight of 22 sum-rank codes}
    Let $\C_1\subset {\mathbb{F}_4^t}$ and $\C_2\subset {\mathbb{F}_4^t}$ be two linear codes with parameters $[t,k_1,d_1]$ and $[t,k_2,d_2]$, respectively. Then a binary linear sum-rank code $SR(\C_1,\C_2)$ with block length $t$, matrix size $2\times 2$ can be constructed explicitly. For any codeword ${\bf c}_1\in \C_1$ and ${\bf c}_2\in \C_2$. Set
    \[I=supp({\bf c}_1)\cap supp({\bf c}_2).\]
    Then
    \[wt_{sr}({\bf c})=2wt_H({\bf c}_1)+2wt_H({\bf c}_2)-3|I|,\]
    where ${\bf c}=(M_1,M_2,\dots, M_t)$, and $M_i$ denote the matrix  corresponding to  $q$-polynomial $f_i(x)={\bf c}_{1i}x+{\bf c}_{2i}x^{q},\ i=1,\dots,t$.
\end{lem}

From the above results, we can obtain the quasi-perfect binary sum-rank codes with matrix size $2\times 2$,

\begin{theorem}\label{quasi-perfect sum-rank codes with 22}
    Let $\C_1\subset {\mathbb{F}_4^{t}}$ be a trivial linear code with parameters $[t,t-1,2]_4$, $\C_2\subset {\mathbb{F}_4^t}$ be a linear code with parameters $[t,k,4]_4$ and covering radius two. Then a binary quasi-perfect sum-rank code $SR(\C_1,\C_2)$ with block length $t$, matrix size $2\times 2$, minimum distance $4$ and covering radius two  is constructed explicitly.
\end{theorem}

\begin{proof}
    We first prove that this sum-rank code $SR(\C_1,\C_2)$ is quasi-perfect. We consider  $\C_1$ as a covering code with radius one. Therefore, it is obvious that there exists a codeword ${\bf c}_1\in \C_1$, such that $wt_{H}({\bf v}_1-{\bf c}_1)=2$ for any vector ${\bf v}_1\in {\mathbb{F}_{q^m}^t}$, and these two nonzero positions can be at any position in $\{1,2,\dots,t\}$. Moreover, $\C_2$ has covering radius two in the Hamming metric. Hence, for any vector ${\bf v}_2\in {\mathbb{F}_{q^m}^t}$ , there exists a codeword ${\bf c}_2\in \C_2$ such that $wt_H({\bf v}_2-{\bf c}_2)=2$. Without loss of generality, we assume that $j_1$-th and $j_2$-th position of vector ${\bf v}_2-{\bf c}_2$ is not zero, denoted by $({\bf v}_2-{\bf c}_2)_{j_i}\neq 0,i=1,2$. We assume that $({\bf v}_1-{\bf c}_1)_{j_i}\neq 0, i=1,2$. Therefore, for any ${\bf v}=(M_1,M_2,\dots,M_t)\in {\mathbb{F}_q^{(2,2),\dots,(2,2)}}$, there exists a codeword ${\bf c}=(N_1,N_2,\dots,N_t)\in SR(\C_1,\C_2)$, where $M_i$ and $N_i$ denote the matrix corresponding to  $q$-polynomial
    \[f_i(x)={\bf v}_{1i}\phi(x)+{\bf v}_{2i}\phi(x^{q}),\ g_i(x)={\bf c}_{1i}\phi(x)+{\bf c}_{2i}\phi(x^{q}),  i=1,\dots,t\]
    respectively, such that ${\bf v}-{\bf c}$ has two nonzero positions. By Lemma \ref{weight of 22 sum-rank codes}, since $supp({\bf v}_1-{\bf c}_1)=supp({\bf v}_2-{\bf c}_2)$, we have
    \[
    wt_{sr}({\bf v}-{\bf c})=2wt_H({\bf v_1}-{\bf c_1})+2wt_H({\bf v}_2-{\bf c}_2)-3|I|=2,\ R_{sr}(SR(\C_1,\C_2))\leq 2.
    \]
    Therefore, all balls with radius two centered at all codewords in $SR(\C_1,\C_2)$ cover the whole space, then $d_{sr}(SR(\C_1,\C_2))=4$ and $SR(\C_1,\C_2)$ is a quasi-perfect sum-rank code.
\end{proof}

\begin{remark}
    From Theorem \ref{quasi-perfect sum-rank codes with 22}, many quasi-perfect binary sum-rank codes with the matrix size $2\times 2$ and the minimum sum-rank distance four can be constructed explicitly from these quasi-perfect codes over ${\mathbb{F}_4}$ constructed in \cite{Giulietti2007}.
\end{remark}

\section{Constructions of distance-optimal $q$-ary sum-rank codes }

In this section, we give some constructions of distance-optimal $q$-ary codes in the sum-metric Firstly, we recall some basic concepts on cyclic codes. A linear code $\C$ is called cyclic if $(c_0,c_1,\dots,c_{n-1})\in \C$, then $(c_{n-1},c_0,\dots,c_{n-2})\in \C$. A codeword ${\bf c}$ in a cyclic code is identified with a polynomial ${\bf c}(x)=c_0+c_1x+\cdots+c_{n-1}x^{n-1}\in {\mathbb{F}_q}[x]/(x^n-1)$. Every cyclic code is a principal ideal in the ring ${\mathbb{F}_q}[x]/(x^n-1)$ and generated by a factor $g(x)$ of $x^n-1$. Let $n$ be a positive integer satisfying ${\rm gcd}(n,q)=1$ and ${\mathbb{Z}}_n={\mathbb{Z}}/n{\mathbb{Z}}=\{0,1,\dots,n-1\}$ be the residue classes modulo $n$. A subset $C_i$ of ${\mathbb{Z}}_n$ is called a $q$-cyclotomic coset if $C_i=\{i,iq,\dots,iq^{\ell-1}\}$, where $i\in {\mathbb{Z}}_n$ and $\ell$ is the smallest positive integer such that $iq^{\ell}\equiv i\pmod n$. Then each $q$-cyclotomic coset modulo $n$ corresponds to an irreducible factor of $x^n-1$ over ${\mathbb{F}_q}[x]$. A
generator polynomial of a cyclic code is the product of several irreducible factors of $x^n-1$. The defining set of a cyclic code generated by $g(x)$ is the the following set
\[
T=\{i\,:\, g(\beta^i)=0,\}
\]
where $\beta$ is a primitive $n$-th unity root of ${\mathbb{F}_{q^m}}$, $m={\rm ord}_n(q)$. Hence the defining set of a cyclic code is the disjoint union of several $q$-cyclotomic cosets. The famous BCH bound asserts that if there are $\delta-1$ consecutive elements in the defining set of a cyclic code, then the minimum distance of this cyclic code is at least $\delta$, see \cite{MacWilliams1977,Huffman2003,Vanlint1999}.

We consider a class of distance-optimal cyclic codes with minimum distance four. Let $q$ be a prime power with $q\geq 4$ and $n=\frac{q^m-1}{\lambda}$, where $m$ is a positive integer and $\lambda$ is a divisor of $q^m-1$. Then every $q$-cyclotomic coset in ${\mathbb{Z}}_n$ has at most $m$ elements.  Then we can construct a class of distance-optimal cyclic codes given as follows:

\begin{theorem}\label{distance-optimal cyclic code in the Hamming metric}
    Let $q\geq 4$ be a prime power, $m$ be a positive integer. If $\lambda$ is a divisor of $q^m-1$ satisfying
    \[\lambda< \frac{q-1}{\sqrt{2q(1+\epsilon)}},\]
    where $\epsilon$  is an arbitrary small positive real number. Then a distance-optimal cyclic code with parameters $[\frac{q^m-1}{\lambda},\frac{q^m-1}{\lambda}-2m-1,4]_q$ is constructed, when $m$ is sufficiently large.
\end{theorem}

\begin{proof}
    We consider the defining set $T=C_0\cup C_1\cup C_2$, then $T$ has at most $2m+1$ elements. Since $0,1,2$ are in $T$,  a cyclic code with parameters $[\frac{q^m-1}{\lambda}, \geq \frac{q^m-1}{\lambda}-2m-1,\geq 4]$ is constructed. It is known that $\lambda< \frac{q-1}{\sqrt{2q(1+\epsilon)}}$, and $C_1=\{1,q,\dots,q^{\ell-1}\}$, where $\ell$ is smallest positive integer such that $q^{\ell}\equiv 1 \pmod n$. It is easy to verify that $\ell=m$ if $\lambda< \frac{q-1}{\sqrt{2q(1+\epsilon)}}$, therefore $|C_1|=m$. Similarly, $|C_2|=m$, then a cyclic code with parameters $[\frac{q^m-1}{\lambda}, \frac{q^m-1}{\lambda}-2m-1,\geq 4]$ is constructed.

    Next, we prove that this class of cyclic codes is distance-optimal with minimum distance $4$. It is sufficient to prove that
    \[V_H(q,\left\lfloor \frac{d_H}{2}\right\rfloor )>q^{2m+1}.\]
    It is known that the volume of the ball with radius $2$ in the Hamming metric space ${\mathbb{F}_q^n}$ is
    \[
    V_H(q,2)=\sum\limits_{i=0}^2\left(\begin{array}{c}
    n \\
    i
    \end{array}\right)(q-1)^i>\frac{n(n-1)(q-1)^2}{2}.
    \]
    Since $\lambda< \frac{q-1}{\sqrt{2q(1+\epsilon)}}$, we have
    \[
    \frac{n(n-1)(q-1)^2}{2}>n(n-1)\lambda^2(1+\epsilon)q=q^{2m+1}(1-\frac{1}{q^m})(1-\frac{\lambda+1}{q^m})(1+\epsilon)>q^{2m+1},
    \]
    when $m$ is sufficiently large. Hence, a distance-optimal cyclic code with parameters $[\frac{q^m-1}{\lambda}, \frac{q^m-1}{\lambda}-2m-1, 4]$ is constructed.

\end{proof}

\begin{remark}
    If $\lambda=1$, we can construct an infinite family of distance-optimal cyclic  codes with parameters $[q^m-1,q^m-1-2m-1,4]$ for any prime power q. These distance-optimal cyclic codes have
    the same parameters as distance-optimal cyclic codes in \cite{Ding2013,Wu2023,Yuan2006}. In addition, we can obtain an infinite family of distance-optimal cyclic codes with parameters $[\frac{5^{2m}-1}{3},\frac{5^{2m}-1}{3}-2m-1,4]$. It seems this is a family of distance-optimal cyclic codes with new parameters.
\end{remark}

\begin{theorem}
    A distance-optimal ternary cyclic code with parameters $[3^m-1,3^m-2m-2,4]$ and defining set $C_0\cup C_1\cup C_5$ is constructed. A distance-optimal quinary cyclic code with parameters $[5^m-1,5^m-2m-2,4]$ and defining set $C_0\cup C_1\cup C_3$ is constructed.
\end{theorem}

\begin{proof}
    It is known that the minimum distance is at least $4$ if $0,1,3,5$ are in the defining set of $q$-ary cyclic codes by Boston bound \cite{Boston2001,Zeh2012}. Then  these two families of cyclic codes have minimum distance at least $4$. Similar to Theorem \ref{distance-optimal cyclic code in the Hamming metric}, we can prove that they are distance-optimal.
\end{proof}

\subsection{Distance-optimal $q$-ary codes in the sum-rank metric with matrix size $s\times s$}

In this subsection, we construct distance-optimal sum-rank codes with matrix size $s\times s$ from distance-optimal cyclic codes in the Hamming metric.

\begin{theorem}\label{distance-optimal cyclic sum-rank code with matrix ss}
    Let $q$ be a prime power, $s$ and $m$ be fixed positive integers. If $\lambda$ is a divisor of $q^{sm}-1$ satisfying
    \[\lambda<\sqrt{\frac{q^s-1}{2(q-1)^2(1+\epsilon)}},\]
    where $\epsilon$ is an arbitrary small positive real number. Then a distance-optimal code in the sum-rank metric with block length $t=\frac{q^{sm}-1}{\lambda}$, matrix size $s\times s$, the cardinality $q^{s^2t-s(2m+3)}$ and minimum sum-rank distance 4 is constructed.
\end{theorem}

\begin{proof}
    Let $\C_1$ be a cyclic code with parameters $[t,t-2m-1,4]_{q^s}$ with defining set $T=C_0\cup C_1\cup C_2$ given in Theorem \ref{distance-optimal cyclic code in the Hamming metric}, $\C_2,\C_3$ be trivial cyclic codes with parameters $[t,t-1,2]_{q^s}$ and $\C_i,i=4,\dots,s$ be trivial cyclic codes with parameters $[t,t,1]_{q^s}$. Then we have $SR(\C_1,\C_2,\dots,\C_s)$ is a sum-rank code and it has block length  $t=\frac{q^{sm}-1}{\lambda}$ and
    \[d_{sr}(SR(\C_1,\dots,\C_s))\geq {\rm min}\{d_1,2d_2,\dots,sd_s\}=4,\ {\rm dim}_{\mathbb{F}_{q}}(SR(\C_1,\dots,\C_s))=s^2t-s(2m+3).\]
    Therefore, $|SR(\C_1,\dots,\C_s)|=q^{s^2t-s(2m+3)}$. Then we prove this cyclic sum-rank code is distance-optimal with minimum distance 4, it is sufficient to prove that
    \[V_{sr}(q,2)>q^{s(2m+3)}.\]
    By Lemma \ref{volume of ball with radius 2 in the sum-rank metric}, we have
    \[V_{sr}(q,2)\geq \frac{t(t-1)(q^s-1)^4}{2(q-1)^2}>\frac{(t-1)^2(q^s-1)^4}{2(q-1)^2}.\]
    It is sufficient to prove that
    \[\lambda<\sqrt{\frac{(q^s-1)(1-\frac{1+\lambda}{q^{sm}})(1-\frac{1}{q^s})^3}{2(q-1)^2}}<\sqrt{\frac{q^s-1}{2(q-1)^2(1+\epsilon)}},
    \]
    when $m$ is sufficiently large. The conclusion follows immediately.
\end{proof}

It is obvious that infinitely many families of distance-optimal cyclic sum-rank codes
can be obtained from Theorem \ref{distance-optimal cyclic sum-rank code with matrix ss}. If $s_1<s_2$, and $\lambda$ is a divisor of $q^{s_2m}-1$,  we consider cyclic codes $\C_i,i=1,2,\dots, s_1$ over ${\mathbb{F}_{q^{s_2}}}$ and sum-rank codes $SR(\C_1,\dots,\C_{s_1})$. Then we can derive the following results.

\begin{cor}
    Let $q$ be a prime power, $s_1<s_2$ and m be fixed positive integers. If $\lambda$ is a divisor of $q^{s_2m}-1$ satisfying
    \[
    \lambda<\frac{q^{s_1}-1}{q-1}\sqrt{\frac{1}{2(1+\epsilon)q^{s_2}}}.
    \]
    Then an infinite family of distance-optimal codes of the block length $t=\frac{q^{s^2m}-1}{\lambda}$, the matrix size $s_1\times s_2$ and the minimum distance 4 is constructed.
\end{cor}

\subsection{Distance-optimal $q$-ary sum-rank codes with matrix size $2\times 2$}

In this subsection, we give the construction of distance-optimal sum-rank codes with matrix size $2\times 2$. Firstly, we have the following known results.

\begin{lem}\cite[Theorem 4.5.6]{Huffman2003}\label{HT bound}
    Let $\C$ be a cyclic code of length $n$ over ${\mathbb{F}_q}$ with defining set $T$. Let $A$ be a set of $\delta-1$ consecutive elements of $T$ and $B=\{jb \pmod n\,:\, 0\leq j\leq s\}$, where ${\rm gcd}(b,n)<\delta$. If $A+B\subseteq T$, then the minimum weight $d$ of $\C$ satisfies $d\geq \delta+s$.
\end{lem}

Let $q$ be a prime power, $n=q^2-1$ and the defining set $T=C_0\cup C_1 \cup C_{q+1}$. We consider the $q$-cyclotomic cosets in ${\mathbb{Z}_n}$ and cyclic code $\C$ with defining set $T$, there are four elements in $T$, and $T=\{0,1,q,q+1\}$. Set $A=\{0,1\}$ and $B=\{0,q\}$. Then by Lemma \ref{HT bound}, $d(\C)\geq 4$.

\begin{theorem}\label{distance-optimal sum-rank code with matrix size 22}
    Let $q$ be a prime power, a distance-optimal sum-rank code with block length $t=q^4-1$, matrix size $2\times 2$, and minimum sum-rank distance 4 is constructed.
\end{theorem}

\begin{proof}
    Let $\C_1$ be a cyclic code with parameters $[q^4-1,q^4-5,4]_{q^2}$ constructed as above, $\C_2$ be a trivial cyclic code with parameters $[q^4-1,q^4-2,2]_{q^2}$. Then $SR(\C_1,\C_2)\subset {\mathbb{F}_q^{(2,2),\dots,(2,2)}}$ is a cyclic sum-rank code with block length $t$ and matrix size $2\times 2$, and
    \[d_{sr}\geq {\rm min}\{d_1,2d_2 \}=4,\ {\rm dim}_{\mathbb{F}_q}(SR(\C_1,\C_2))=4t-10.\]
    It is sufficient to prove $V_{sr}(q,2)>q^{10}$ if $SR(\C-1,\C_2)$ is distance-optimal. By Lemma \ref{volume of ball with radius 2 in the sum-rank metric},
    \[
    V_{sr}(q,2)\geq \frac{t(t-1)(q^2-1)^4}{2(q-1)^2}>q^{13}.
    \]
    Therefore, $SR(\C_1,\C_2)$ is distance-optimal sum-rank code.
\end{proof}

\begin{remark}
    The sum-rank code constructed in Theorem \ref{distance-optimal sum-rank code with matrix size 22} has Singleton defect  $2(2t-4+1)-(4t-10)=4$. Infinitely many distance-optimal sum-rank codes, which are close to Singleton-like bound and have much larger $q^4-1\gg q-1$  block lengths, are constructed. These sum-rank codes are next best possibility to the almost MSRD codes with the Singleton defect 2. 
\end{remark}

Next we present an almost MSRD code with block length up to $q^2$,  matrix size $2\times 2$.
\begin{theorem}
    Let $q$ be a prime power, an almost MSRD code with block length $t$ up to $q^2$, matrix size $2\times 2$, minimum sum-rank distance four is constructed explicitly.
\end{theorem}
\begin{proof}
    Let $\C_1$ be a Reed-Solomon code with parameters $[t,t-3,4]_{q^2}$, where $t\leq q^2$, $\C_2$ be a trivial code with parameters $[t,t-1,2]_{q^2}$. Then $SR(\C_1,\C_2)$ is a sum-rank code with block length $t$, matrix size $2\times 2$,
    \[d_{sr}(SR(\C_1,\C_2))\geq 4,\ {\rm and}\ {\rm dim}_{\mathbb{F}_q}(SR(\C_1,\C_2))=2(t-3+t-1)=4t-8.\]
    We have $d_{sr}(SR(\C_1,\C_2))=4$ since
    \[
    V_{sr}(q,2)\geq \frac{t(t-1)(q^2-1)^4}{2(q-1)^2}>q^8
    \]
    and the Singleton defect of $SR(\C_1,\C_2)$ is $2(2t-d_{sr}+1)-2(2t-4)=2$. Then $SR(\C_1,\C_2)$ is an almost MSRD code.
\end{proof}
\begin{remark}
    Comparing MSRD codes with the block length up to $q-1$ constructed in \cite{Neri2023,Martinez2018}, these almost MSRD codes have larger block lengths up to $q^2$. One of the challenging problem is to construct almost MSRD codes with the block lengths up to $q^2$ and larger minimum sum-rank distances.
\end{remark}

\section{Plotkin sum of sum-rank codes}

As in the case of the matrix size $1\times 1$ sum-rank codes (codes in the Hamming metric), the Plotkin sum of sum-rank codes of the matrix size $n\times m,n\leq m$ can be defined and good sum-rank codes can be obtained.

Let $\C_1$ and $\C_2$ be two sum-rank codes over ${\mathbb{F}_q^{(n,m),\dots,(n,m)}}$ with block length $t$, minimum sum-rank distances $d_1$ and $d_2$, and dimensions $k_1$ and $k_2$. Then we define their Plotkin sum-rank code as
\[{\rm Plotkin(\C_1,\C_2)}=\{({\bf c}_1|{\bf c}_1+{\bf c}_2)\,:\, {\bf c}_1\in \C_1,{\bf c}_2\in \C_2\}.\]

\begin{theorem}
    Let $\C_1,\C_2$ be two sum-rank codes given as above. Then the dimension of their Plotkin sum is $k_1+k_2$ and the minimum sum-rank distance of their Plotkin sum is ${\rm min}\{2d_1,d_2\}$.
\end{theorem}

\begin{proof}
    It is obvious that the dimension of ${\rm dim}_{\mathbb{F}_q}({\rm Plotkin}(\C_1,\C_2))=k_1+k_2$. If ${\bf c}_2\neq {\bf 0}$, then
    \[
    wt_{sr}({\bf c}_1|{\bf c}_1+{\bf c}_2)=wt_{sr}({\bf c}_1)+wt_{sr}({\bf c}_1+{\bf c}_2)\geq wt_{sr}({\bf c}_1+{\bf c}_2-{\bf c}_1)=wt_{sr}({\bf c}_2)\geq d_2.
    \]
    If ${\bf c}_2={\bf 0},$ then
    \[
    wt_{sr}({\bf c}_1|{\bf c_1})=2wt_{sr}({\bf c}_1)\geq 2d_1.
    \]
    The conclusion follows immediately.
\end{proof}

In Theorem \ref{distance-optimal cyclic sum-rank code with matrix ss}, we get an infinite family of distance-optimal binary
sum-rank codes $\C$ with the block length  $t=q^{sm}-1$, the matrix size $s\times s$, the dimension $s^2t-s(2m+3)$ and the minimum sum-rank distance four. Then we can obtain new distance-optimal $2$-ary sum-rank codes with the block length $2t = 2(2^{sm}-1)$.

\begin{theorem}
Let $q=2^s$, s and m be two positive integers, $t=2^{sm}-1$. Then a distance-optimal sum-rank code  can be constructed by Plotkin sum with block length $2t$, matrix size $s\times s$ and  minimum sum-rank distance four.
\end{theorem}

\begin{proof}
    Let $\C_1'$ be a trivial cyclic code with parameters $[t,t-1,2]_{2^s}$, $\C_i',i=2,\dots,s$ be trival cyclic codes with parameters $[t,t,1]_{2^s}$, and $\C_1=SR(\C_1',\C_2',\dots,\C_s')\subset {\mathbb{F}_2^{(s,s),\dots,(s,s)}}$ be a sum-rank code with block length $t$, matrix size $s\times s$.  Then we have
    \[
    d_{sr}(\C_1)\geq 2,\ {\rm and }\ {\rm dim}_{\mathbb{F}_2}(\C_1)=s^2t-s.
    \]
    Let $\C_2$ be a sum-rank code constructed in Theorem \ref{distance-optimal cyclic sum-rank code with matrix ss}. Then we can construct a sum-rank code ${\rm Plotkin}(\C_1,\C_2)$ with block length $2t$, matrix size $s\times s$,
    \[
    d_{sr}({\rm Plotkin}(\C_1,\C_2))\geq 4, \ {\rm and}\ {\rm dim}_{\mathbb{F}_2}({\rm Plotkin}(\C_1,\C_2))=2s^2t-s-s(2m+3).
    \]
    It is distance-optimal sum-rank code since
    \[
    V_{sr}(2,2)\geq \frac{2t(2t-1)(2^s-1)^4}{2(2-1)^2}=2^{2sm+4s+1}(1-\frac{1}{2^{sm}})(1-\frac{3}{2^{sm+1}})(1-\frac{1}{2^s})^4\geq 2^{2sm+4s},
    \]
    when $m$ is large. The conclusion follows directly.
\end{proof}

\section{Concluding remarks}

In this paper, we have studied covering codes, quasi-perfect codes and distance-optimal codes in the sum-rank metric. Our main contributions are summarized as follows.
\begin{itemize}
    \item  We derived improved upper bounds on the sizes, covering radii, and block length functions of sum-rank codes. These bounds extend classical results about codes in the Hamming metric to codes in the sum-rank metric.
    As applications, we established several strong Singleton-like bounds that are stronger than the Singleton-like bound for sum-rank codes when the block length is large.
    \item We  construct infinitely many new families of distance-optimal $q$-ary cyclic codes with the minimum distance four and  provided explicit constructions of distance-optimal sum-rank codes with matrix sizes $s\times s$ and $2\times 2$ with minimum sum-rank distance four by using cyclic codes in the Hamming metric.
    \item We presented several families of quasi-perfect sum-rank codes with the matrix size $2\times m$ and $2\times 2$, these quasi-perfect sum-rank codes are distance-optimal automatically, and provided a method to construct binary distance-optimal sum-rank codes by using Plotkin sum of sum-rank codes.
\end{itemize}


\begin{thebibliography}{99}
    \bibitem{Abiad2023}A. Abiad, A. Khramova and A. Ravagnani, ``Eigenvalue bounds for sum-rank-metric codes," IEEE Transactions on Information Theory, vol. 70, no. 7, pp. 4843-4855, Jul. 2024.


    \bibitem{Alfarano2022}G. N. Alfarano, F. J. Lobillo, A. Neri, and A. Wachter-Zeh, ``Sum-rank product codes and bounds on the minimum distance," Finite Fields and Their Applications, vol. 80, 102013, 2022.

    \bibitem{BBB} M. Bonini, M. Borello and E. Byrne, The geometry of covering codes in the sum-rank metric, Designs, Codes and Cryptography, vol. 93, pp. 2993-3009, 2025.

    \bibitem{Brualdi1989}R. A. Brualdi, V. S. Pless and R. M. Wilson, ``Short codes with a given radius," IEEE Transactions on Information Theory, vol. 35, no. 1,  pp. 99-109, 1989.

    \bibitem{Boston2001}N. Boston, ``Bounding minimum distances of cyclic codes using algebraic geometry," Electron. Notes Discr. Math., vol. 6, pp. 385-394, 2001.

   \bibitem{Byrne2017} E. Byrne and A. Ravagnani, "Covering radius of matrix codes endowed with rank-metric," SIAM Journal on Discrete Mathematics, vol. 31, no. 2, pp. 927-944, 2017.

    \bibitem{Byrne2021} E. Byrne, H. Gluesing-Luerssen and A. Ravagnani, ``Fundamental properties of sum-rank-metric codes," IEEE Transactions on Information Theory, vol. 67, no. 10, pp. 6456-6475, 2021.

    \bibitem{Byrne2022}E. Byrne, H. Gluesing-Luerssen and A. Ravagnani, ``Anticodes in the sum-rank metric," Linear Algebra and its Applications, vol. 643, pp. 80-98, 2022.

    \bibitem{Cohen1997}G. D. Cohen, I. Honkala, S. Litsyn and A. Lobstein, Covering codes, North-Hollan
    Math, Libarary, Elsecier, 1997.

    \bibitem{Chen2023}H. Chen, ``New explicit linear sum-rank-metric codes," IEEE Transactions on Information Theory, vol. 69, no. 10, pp. 6303-6313, 2023.

    \bibitem{Chen2025}H. Chen, Z. Cheng, Y. Qi, ``Construction and fast decoding of binary linear sum-rank-metric codes,'' IEEE Transactions on Information Theory, vol. 71, no. 12, pp. 9319-9329, 2025.




    \bibitem{yqchen} Y. Chen, Z. Cheng and H. Chen, Goppa type sum-rank codes, Cryptography and Communications, vol. 17, pp. 421-431, 2025.

    \bibitem{Cai2022}H. Cai, Y. Miao, M. Schwartz and X. Tang, ``A construction of maximally revoerable codes with order-optimal field size," IEEE Transactions on Information Theory, vol. 68, no. 1, pp. 204-212, 2022.

    \bibitem{Camp2022}E. Camps-Moreno, E. Gorla, C. Landolina, E. L. Garc\'{\i}a, U. Mart\'{\i}nez-Pe\~{n}as,  F. Salizzoni, ``Optimal anticodes, MSRD codes and the generalized weights in the sum rank metric," IEEE Transactions on Information Theory, vol. 68, no. 6, pp. 3806-3822, 2022.


    \bibitem{Ding2013}C. Ding and T. Helleseth, ``Optimal ternary cyclic codes from monomials," IEEE Transactions on Information Theory, vol. 59, no. 9, pp. 5898-5904, 2013.

    \bibitem{Davydov2022}A. A. Davydov, S. Marcugini and F. Pambianco, ``Upper bounds on the length function for covering codes with covering radius $R$ and codimension $tR+1$," Advances in Mathematics of Communications, vol. 17, no. 1, pp. 98-118, Jan., 2022.

    \bibitem{Etzion2005}T. Etzion and B. Mounits, ``Quasi-perfect codes with small distance," IEEE Transactions on Information Theory, vol. 51, no. 11, pp. 3938-3946, 2005.


    \bibitem{Giulietti2007}M. Giulietti and F. Pasticci, ``Quasi-perfect linear codes with minimum distance 4," IEEE Transactions on Information Theory, vol. 53, no. 5, pp. 1928-1934, 2007.

    \bibitem{Huffman2003} W. C. Huffman and V. Pless, Fundamentals of error-correcting codes, Cambridge University Press, Cambridge, U. K., 2003.

    \bibitem{Lu2005}H. F. Lu and P. V. Kumar, ``A unified construction of space-time codes with optimal rate-diversity tradeoff," IEEE Transactions on Information Theory, vol. 51, no. 5, pp. 1709-1730, 2005.

    \bibitem{MacWilliams1977}F. J. MacWilliams and N. J. A. Sloane, The Theory of error-correcting codes, 3rd Edition, North-Holland Mathematical Library, vol. 16. North-Holland, Amsterdam, 1977.

    \bibitem{Martinez2018}U. Mart\'{\i}nez-Pe\~{n}as, ``Skew and linearized Reed-Solomon codes and maximal sume rank distance codes over any division ring," Journal of Algebra, vol. 504, pp. 587-612, 2018.



    \bibitem{UMart2019}U. Mart\'{\i}nez-Pe\~{n}as, ``Hamming and simplex codes from the sum-rank metic," Designs, Codes and Cryptography, vol. 88, pp. 1521-1539, 2019.


    \bibitem{MP24} U. Mart\'{\i}nez-Pe\~{n}as and S. Puchinger, Maximum sum-rank distance codes over finite chain rings, IEEE Transactions on Information Theory, vol. 70, no. 6, pp. 3878-3890, 2024.
		

    \bibitem{Neri2022}A. Neri, ``Twisted linearized Reed-Solomon codes: A skew polynomial framework," Journal of Algebra, vol. 609, pp. 792-839, Jun. 2022.

    \bibitem{Neri2023}A. Neri, P. Santonastaso and F. Zullo, ``The geometry of one-weight codes in the sum-rank metric,''  Journal Combinatorial Theory, Ser.A, vol. 194, 105703, 2023.

    \bibitem{Napp2018}D. Napp, R. Pinto and V. Sidorenko, ``Concatenation of covolutional codes and rank metric codes for muti-shot network coding," Designs, Codes and Cryptography, vol. 86, no. 2, pp. 303-318, 2018.

    \bibitem{Ott2021}C. Ott, S. Puchinger and M. Bossert, ``Bounds and genericity of sum-rank-metric codes," XVII International Symposium on Problems of Redundancy in Information and Control Systems, REDUNDANCY 2021.

    \bibitem{Ott2022}C. Ott, H. Liu and A. Wachter-Zeh, ``Covering properties of sum-rank metric codes," 58th Annual Allerton Conference on Communication, Control and Computing, 2022.

    \bibitem{Reed1960}I. Reed and G. Solomon, ``Polynomial codes over certain finite fields," Journal of SIAM, vol. 8, pp. 300-304, 1960.

    \bibitem{Singleton1964} R. Singleton, ``Maximum distance $q$-ary codes,'' IEEE Transactions on Information Theory, vol. 10, no. 2, pp. 116-118, 1964.

    \bibitem{Shehadeh2022}M. Shehadeh and F. R. Kschischang, ``Space-time codes from sum-rank codes," IEEE Transactions on Information Theory, vol. 68, no. 3, pp. 1614-1637, 2022.

    \bibitem{Vanlint1999}J. H. van Lint, Introduction to the coding theory, GTM 86, Third and Expanded Edition, Springer, Berlin, 1999.

    \bibitem{Wu2023}G. Wu, H. Liu and Y. Zhang, ``Several classes of optimal $p$-ary cyclic codes with minimum distance four," Finite Fields and Their Applications, vol. 92, 102275, 2023.

    \bibitem{Yuan2006}J. Yuan, C. Carlet and C. Ding, ``The weight distribution of a class of linear codes from
    perfect nonlinear functions," IEEE Transactions on Information Theory, vol. 52, no. 2, pp. 712-717, 2006.

    \bibitem{Zeh2012}A. Zeh, A. Wachter-Zeh and S. Bezzateev, ``Decoding cyclic codes up to a new bound on the minimum distance," IEEE Transactions on Information Theory, vol. 58, no. 6, pp. 3951-3960, 2012.
	
	
\end{thebibliography}
\end{document}